\documentclass{article}
\usepackage{ijcai24}

\usepackage{times}
\usepackage{soul}
\usepackage{url}
\usepackage[hidelinks]{hyperref}
\usepackage[utf8]{inputenc}
\usepackage[small]{caption}
\usepackage{graphicx}
\usepackage{amsmath}
\usepackage{amsthm}
\usepackage{booktabs}
\usepackage{algorithm}
\usepackage{algorithmic}
\usepackage[switch]{lineno}
\usepackage{caption} 
\usepackage{amssymb}
\usepackage{nicefrac}       
\usepackage{microtype}      
\usepackage{xcolor}         
\usepackage{multirow}
\usepackage{subcaption}
\usepackage{cleveref}
\usepackage{placeins}

\usepackage{thmtools,thm-restate} 

\newtheorem{definition}{Definition}

\title{Protecting Split Learning by Potential Energy Loss}

\author{
    Fei Zheng$^1$
    \and
    Chaochao Chen$^1$
    \and
    Lingjuan Lyu$^2$
    \and
    Xinyi Fu$^3$
    \and\\
    Xing Fu$^3$
    \and
    Weiqiang Wang$^3$    
    \and
    Xiaolin Zheng$^1$\thanks{Corresponding author}
    \and
    Jianwei Yin$^1$
    \affiliations
    $^1$College of Computer Science and Technology, Zhejiang University \\
    $^2$Sony AI \\
    $^3$Ant Group
    \emails
    \{zfscgy2, zjuccc\}@zju.edu.cn, lingjuan.lv@sony.com, fxy122992@antgroup.com, \\
    \{fux008, wang.weiqiang\}@gmail.com,
    xlzheng@zju.edu.cn, zjuyjw@cs.zju.edu.cn
}

\usepackage{bibentry}

\begin{document}

\maketitle

\begin{abstract}
As a practical privacy-preserving learning method, split learning has drawn much attention in academia and industry.
However, its security is constantly being questioned since the intermediate results are shared during training and inference.
In this paper, we focus on the privacy leakage from the forward embeddings of split learning.
Specifically, since the forward embeddings contain too much information about the label, the attacker can either use a few labeled samples to fine-tune the top model or perform unsupervised attacks such as clustering to infer the true labels from the forward embeddings.
To prevent such kind of privacy leakage, we propose the potential energy loss to make the forward embeddings more `complicated', by pushing embeddings of the same class towards the decision boundary.
Therefore, it is hard for the attacker to learn from the forward embeddings.
Experiment results show that our method significantly lowers the performance of both fine-tuning attacks and clustering attacks.

\end{abstract}

\section{Introduction}
Split learning~\cite{vepakomma2018split_health,gupta2018distributed_learning} is a practical method for privacy-preserving machine learning on distributed data.
By splitting the model into multiple parts (sub-models), split learning allows different parties to keep their data locally and only to share the intermediate output of their sub-models.
Compared to cryptographic methods like \cite{mohassel2017secureml,rathee2020cryptflow2,huangzhicong2022cheetah}, split learning is much more efficient in computation and communication.
To date, split learning has been applied in multiple fields, e.g., graph learning~\cite{ccc2022vfgnn}, medical research~\cite{jeong2021split_medical}, and the internet of things~\cite{yusuke2020split_power}.
%

To perform split learning, the model should be split into multiple parts.
Without loss of generality, we suppose the model is split into two parts, i.e., the bottom model $M_b$ and the top model $M_t$, held by the feature owner (Alice) and the label owner (Bob), respectively.
As shown in \Cref{fig:overview} (top part), 
during the forward pass, Alice feeds the input feature $X$ to the bottom model to get the forward embedding $Z = M_b(X)$, then sends $Z$ to Bob.
Bob feeds $Z$ to the top model and gets the prediction $\hat Y = M_t(Z)$.
As for the backward pass, Bob computes the gradients on the loss $\partial L/\partial M_t$ and $\partial L/\partial Z$.
He uses the former one to update $M_t$ and sends the latter one to Alice.
Alice then computes the gradient with respect to $M_b$ and updates its parameters. 
From the above description, we can see that split learning is very straightforward with low computation and communication overhead compared with cryptographic methods.

\begin{figure}[t]
    \centering
    \includegraphics[width=\linewidth]{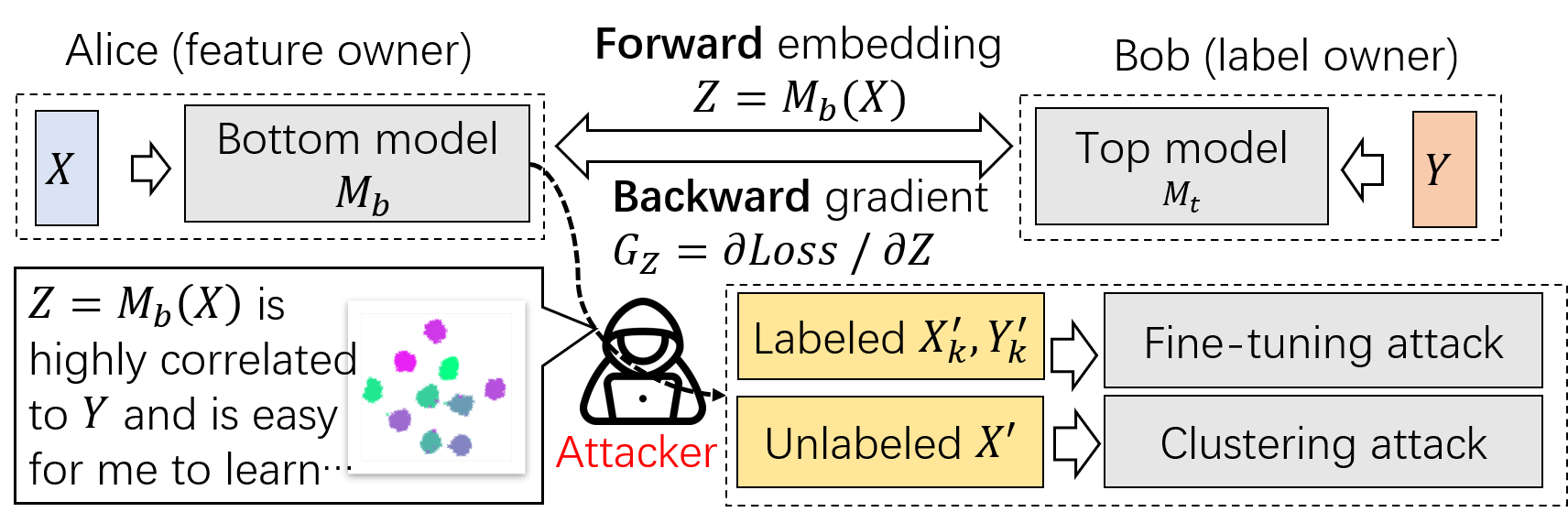}
    \caption{Attacks on the forward embeddings of split learning.}
    \label{fig:overview}
\end{figure}

However, the price for efficiency is privacy.
Many previous studies have investigated the privacy leakage of input features in split learning caused by the \textit{exchange of intermediate results} \cite{abuadbba2020split_cnn,pasquini2021inference_attack}.
Instead, we focus on the privacy leakage caused by the \textit{trained split model itself} in classification tasks, which has been demonstrated by \cite{fu2022vertical_federated_label,sunjiankang2022split_learning_label}.
Consider the two-party split learning scenario described in the previous paragraph.
As shown in \Cref{fig:overview} (bottom part),
if $M_b$ is trained well, it gains the ability to separate samples of different classes and cluster the samples of the same class, in other words, $Z$ becomes `meaningful' and highly correlated with the label.
Hence, Alice (or whoever obtains the bottom model) can 1) fine-tune $M_t$ from random initialization with a few labeled samples, or 2) simply perform clustering with enough unlabeled input samples.
In both cases, Alice can learn the complete classification model based on the forward embeddings.
Considering that the complete model and thelabel could be private assets,
such \textit{model completion attack} caused by the forward embedding poses a significant privacy threat to split learning.

The protection of input features in split learning is already studied.
For example, Vepakomma et al. \shortcite{Vepakomma2020nopeek} decorrelate $H_b$ and $X$ by adding distance correlation~\cite{szekeley2007dcor} loss.
This method is empirically successful because $Z$ does not need to contain the majority of $X$'s information---it only needs to contain the part most relevant to the label $Y$.
On the other side, protecting label information is more challenging.
Since $M_t(Z) = \hat Y$ is the model prediction, $Z$ contains all information about $\hat Y$.
While the training target is to make the prediction $\hat Y$ and the real label $Y$ as close as possible, it seems impossible to prevent the attacker from deriving $Y$ from $Z$.
%

%
To solve this problem, in this paper, we view it from a different perspective.
At first glance, the attacker's target seems exactly the same as Bob's training objective, i.e., to learn a mapping from $Z$ to $Y$.
But there is a crucial difference between them.
Alice and Bob perform the learning procedure on the entire training dataset, but on the contrary, the attacker can only access a small number of labeled samples by assumption.
The problem now turns into modifying $Z$'s distribution so that the attacker learns poorly with a few labeled samples while the benign actors learn well with sufficient labeled samples.
We observe that such distribution can be obtained if the embeddings of the same class are distributed near the boundary of the decision region.
However, during vanilla split learning, different-class embeddings tend to separate from each other while same-class embeddings are densely clustered, which is opposite to our purpose.
Our solution is inspired by a well-known physics phenomenon named electrostatic equilibrium, i.e., all the net charges distribute on the surface of the conductor.
Interestingly, this distribution happens to be the desired distribution for $Z$, if we view the forward embeddings as charges.
One reason for this phenomenon is that there is a repulsive force between every pair of \textit{like charges} (charges that have the same sign, e.g., positive charges).
Inspired by this, we propose the \textit{potential energy loss} on the forward embeddings.
The potential energy loss adds a repulsive force between each pair of same-class embeddings.
During training, same-class embeddings are pushed toward the boundary of the decision region, resulting in a large learning error for the attacker.
Therefore, the attacker cannot fine-tune the bottom model with a few labels or cluster unlabeled embeddings.

In summary, we make the following contributions:
\begin{itemize}
    \item We formalize the privacy leakage of the forward embeddings in terms of the learning error of the attacker, and demonstrate that making embeddings lie near the decision boundary enhances privacy.
    \item We propose the potential energy loss on the forward embeddings to push them to their decision boundary, in order to reduce the privacy leakage from the forward embeddings.
    \item We conduct extensive experiments on multiple datasets, showing that our method significantly reduces the attacker's learning accuracy of both fine-tuning attacks and clustering attacks, and performs better than the existing distance correlation approach.
\end{itemize}

\section{Related Work}
\noindent\textbf{Privacy Concerns of Split learning.}
Many studies have demonstrated the privacy concerns of split learning.
Most of them focus on the privacy of input features.
Abuadbba et al. \shortcite{abuadbba2020split_cnn} show that applying split learning to CNN models can be dangerous since the intermediate output is highly correlated to the input.
Luo et al. \shortcite{luo2021inference_attack} propose methods for feature inference attacks of split learning under certain conditions.
As for the privacy of the label data, the attack can be based on either the forward embeddings or the backward gradients.
For \textbf{forward embeddings}:
Sun et al. \shortcite{sunjiankang2022split_learning_label} find that bottom model output (forward embeddings) also leaks label data.
Fu et al. \shortcite{fu2022vertical_federated_label} point out that the attacker can easily fine-tune the top model with a few labeled samples.
For \textbf{backward gradients}:
Li et al. \shortcite{li2021split_learning_label} investigated the label leakage brought by the backward gradients, for example, the norm of different classes' gradients can be different.
Studies \cite{erdougan2022unsplit,sanjay2023exploit_split_learning} use a surrogate model and label to match the backward gradients and reveal the training labels.
While the above studies assume the attacker is \textit{passive}, i.e., he will not tamper with the training procedure but only exploits the information he received,
Pasquini et al. \shortcite{pasquini2021inference_attack}  propose an \textit{active} attack by making the bottom model invertible by modifying the training objective.

\vspace{5pt}
\noindent\textbf{Privacy Protection for Split Learning.}
\label{sec:RW-PPSL}
Aside from relatively expensive cryptographic-based methods such as privacy-preserving machine learning \cite{mohassel2017secureml,rathee2020cryptflow2,huangzhicong2022cheetah} 
or partial cryptographic split learning methods \cite{fufangcheng2022blindFL,chenchaochao2020codesign},
non-cryptographic methods mostly protect privacy by perturbing the forward embeddings or backward gradients.
\textbf{Perturbing forward embeddings}:
Vepakomma et al. \shortcite{Vepakomma2020nopeek} add a distance correlation~\cite{szekeley2007dcor} loss to decorrelate the input features, and similarly, Sun et al. \shortcite{sunjiankang2022split_learning_label} use distance correlation to protect the label.
Duan et al. \shortcite{duanlin2023privascissors} perturb the embedding by minimizing the mutual information with both label and features to protect both, however, it requires a customized model structure.
Chen et al. \shortcite{ccc2022ppcdr} protect the embedding in the recommendation model via differential privacy.
\textbf{Perturbing backward gradients}:
Li et al. \shortcite{li2021split_learning_label} protect the label during training by perturbing the backward gradients, although the forward embeddings could still leak the label information.

\vspace{5pt}
\noindent\textbf{Data-Dependent Generalization Error. }
Most studies on data-dependent generalization error are based on the Rademacher and Gaussian complexity
~\cite{koltchinskii2002empirical,kontorovich2014knn_generalization,leiyunwen2019data-dependent}, 
or the mutual information between the data and the algorithm output~\cite{negrea2019information_generalization,pensia2018generalization_iterative,russo2020overfit}.
%
Different from them, Jin et al. \shortcite{jinpengzhan2020generalization} derived generalization bounds directly from the data distribution, by proposing the so-called cover complexity, 
which is computed from the distances between same-class data points and different-class data points.
It is somewhat related to our work since our method makes the data distribution more `complicated' by pushing the data points to the decision boundary of their class.

\section{Problem: Model Completion Attack}
In this section, we demonstrate and formalize the privacy problem arising from the bottom model in split learning.

\subsection{Leakage from Forward Embeddings}
The hidden embeddings of neural networks are widely studied \cite{rauber2017visualize_hidden,pezzotti2018deepeyes,cantareira2020hidden}.
Through visualization and other techniques, those studies show that the neural network gradually learns to make hidden embeddings of different classes separate, and those of the same class clustered together.
Although this `separation ability' seems to be essential for neural networks and may be the reason why they perform well on various tasks, it also brings security hazards for split learning.
In split learning, the model is split into one (or multiple) bottom model(s) and top model(s), which are held by different parties.
Ideally speaking, any single party can not perform inference tasks, since he only gets a part of the complete model.
Only multiple parties work together, can they make use of the complete model and perform inference tasks.
In other words, the model should be trained and used in a `shared' manner.

However, given the fact that the hidden embeddings of the model are meaningful, the attacker who has the bottom model can either fine-tune the top model with a small number of labeled samples~\cite{fu2022vertical_federated_label}, or just perform clustering on the forward embeddings to infer the labels~\cite{sunjiankang2022split_learning_label}.
Thus, the privacy of split learning is violated.

\subsection{Threat Model}
We consider two threat models, i.e., fine-tuning attack with a few labeled samples and unsupervised clustering attack with massive unlabeled samples.

\vspace{5pt}\noindent\textbf{Fine-tuning Attack.}
We assume that the attacker has access to the trained bottom model $M_b$, along with a few labeled samples $X'_k, Y'_k$ which include $k$ samples for each class.
The attacker also knows the architecture of the top model $M_t$, and performs the model completion attack via training $M_t$ from a random initialization, given $X_k'$ and $Y_k'$, with pre-trained $M_b$ fixed.

\vspace{5pt}\noindent\textbf{Clustering Attack.}
We assume that the attacker has access to the trained bottom model $M_b$, along with massive unlabeled samples $X'$.
To infer the labels of $X'$, the attacker performs clustering algorithms on the forward embedding $Z' = M_b(X')$.

\vspace{5pt}
Notably, we assume \textbf{the attacks are conducted in the inference phase instead of the training phase}, as our method aims to reduce the privacy leakage from forward embeddings.
To prevent leakage from backward gradients during training, one can use non-sensitive data for training or adopt existing approaches such as cryptography-based secure computation described in \Cref{sec:RW-PPSL}.
It is worth noting that, some unsupervised/semi-supervised learning methods can train the complete model with good performance over unlabeled/partially labeled data~\cite{berthelot2019mixmatch,xuyi2021dp_ssl}.
However, they are not relevant to split learning since they do not require any knowledge about the trained bottom model or forward embedding.

\subsection{Problem Formulation}
In order to reduce the aforementioned privacy leakage while maintaining the model performance at the same time, our purpose is to train a split model $M = (M_b, M_t)$ such that the output of $M_b$ is hard for the attacker to learn, while the complete $M$ still maintains a high performance.

\begin{definition}[Bottom model advantage]
The bottom model advantage is the extra advantage obtained by the attacker when he has access to the trained bottom model.
Consider an attack algorithm $\mathcal A$ whose input is the data $D$ and (optionally) the bottom model $M_b$, the bottom model advantage is defined as follows:
\begin{equation}
\label{eq:bottom-adv}
\small
  \begin{split}
    & Adv(M_b;\mathcal A) = \mathbb E_D \left\{ R[\mathcal A(D;null)] - R[\mathcal A(D;M_b)] \right\},
\end{split}
\end{equation}
\end{definition}
where $R[\cdot]$ is an error metric for the attack outcome.
For example, in the fine-tuning attack, $D = (X_k', Y_k')$ is the leaked labeled samples, and $R[\cdot]$ is the error of the fine-tuned model on the test data; in the unsupervised attack, $D = X$ is the unlabeled data, and $R[\cdot]$ is the error of the clustering model on the test data.
We use $\mathcal A(D; null)$ to represent the case that the attacker trains the whole model solely based on the leaked data, without any information about the bottom model or forward embedding (i.e., training from scratch).

\begin{definition}[Perfect protection]
For an attack algorithm $\mathcal A$, if the bottom model $M_b$ satisfies $Adv(M_b;\mathcal A) = 0$, then we say that $M_b$ achieves perfect protection against attack $\mathcal A$, since $M_b$ provides no extra advantage for the attack.
\end{definition}

Thus, our purpose becomes to train a split model $(M_b, M_t)$ such that:
\begin{itemize}
    \item The performance of the complete model on the original task is as high as possible.
    \item The bottom model advantage under fine-tuning attacks and clustering attacks is as small as possible. 
\end{itemize}






\section{Method: Potential Energy Loss}
In this section, we view the privacy leakage of the bottom model as a learning problem for the attacker.
We first study the generalization error when fine-tuning the model with a small number of labeled samples.
By a simplified example, we show that pushing the embeddings of same-class samples toward the decision boundary increases the generalization error.
At the same time, clustering is difficult since the same-class embeddings are no longer close to each other.
Inspired by the electrostatic equilibrium and Coulomb's law, we propose the potential energy loss on the forward embedding, to realize such distribution.
The high-level view of our idea is presented in \Cref{fig:idea}.

\subsection{Learning Error from Data Distribution}
Recall that our goal is to train a split model $(M_b, M_t)$, such that the bottom model provides little advantage to the attacker.
To do this, we consider the attack process to be a learning process on the forward embeddings produced by the bottom model.
We show that when the same-class embeddings are distributed near the boundary of the decision region, the performance of fine-tuning attack is decreased since a small number of samples cannot represent the overall distribution, and a small error on the estimation of decision boundary will cause a large classification error.
Moreover, it naturally prevents clustering attacks since same-class embeddings are no longer close to each other.

\subsubsection{Genralization Error}
We use a simplified example to get some insights into the relationship between the data distribution and generalization error.
Assume that all data points are distributed on the $d$-sphere $\{\mathbf x: \sum\limits_{i=1}^d x_i^2 = 1\}$.
Let the hypothesis set be an arbitrary hemisphere
\begin{equation}
  \mathcal H = \{h: h(\mathbf x) = \text{Sign}[\mathbf w \cdot \mathbf x], \Vert \mathbf w \Vert_2 = 1\}.
\end{equation}
Without loss of generality, we assume the target hypothesis is $f(\mathbf x) = \text{Sign}(x_1)$.
We make the following assumptions:
\begin{itemize}
    \item The probability density of samples only depends on the first dimension $x_1$, i.e., it is isotropic in any other dimensions.
    \item Given a set of positive samples $S = \{\mathbf x_1, ..., \mathbf x_n\}$, the learning algorithm simply outputs the normalized mean of these samples as the parameter of learned hypothesis, i.e., $f^{(S)}(\mathbf x) = \text{Sign}\left[\sum\limits_{i=1}^n \mathbf x_i \cdot \mathbf x \big / \big\lVert\sum\limits_{i=1}^n \mathbf x_i\big\rVert_2 \right]$.
\end{itemize}

Now we want to estimate the generalization error when the learned parameter $\mathbf w = \sum\limits_{i=1}^n \mathbf x_i \big / \big\lVert\sum\limits_{i=1}^n \mathbf x_i\big\rVert_2$ slightly differs from the true parameter $\mathbf e_1$.
Since the distribution is isotropic except in the direction of $\mathbf e_1$, we may assume that $\mathbf w$ lies on the plane expanded by the first two axis, i.e., $\mathbf w = \mathbf e_1 \cos \epsilon + \mathbf e_2 \sin \epsilon$, where $\epsilon$ is a small angle between $\mathbf w$ and $\mathbf e_1$.
The generalization error is 
\begin{equation}
\begin{split}
    \dfrac12 \cdot R[\mathbf w] & = \mathop{\mathbb E}\limits_{\mathbf x \sim \mathcal S} \text{Sign}[x_1] \cdot I[x_1 \cos \epsilon + x_2 \sin \epsilon \le 0] 
    \\
    & = \int_{\substack{x_1 > 0 \\ x_1 \cos \epsilon + x_2 \sin \epsilon \le 0 \\ x_1^2 + ... + x_d^2 = 1}}p(x_1, x_2, ...,x_d)dS 
    \\
    & \le \int_{\substack{x_1^2 + ... + x_d^2 = 1 \\0 < x_1 \le \tan\epsilon}}p(x_1, x_2, ...,x_d)dS 
    \\ 
    & \approx \int_{x_1 = 0}^{\epsilon} p_1(x_1) dx_1 \approx \epsilon p_1(0),
\end{split}
\label{eq:error-estimation}
\end{equation}
where $p$ is the probability density of sample feature, and $p_1$ is the marginal density function of $x_1$.
From \eqref{eq:error-estimation} we can see that with $\epsilon$ fixed, the generalization bound is approximately proportional to the probability mass of the data points falling near the boundary of the target region.

\subsubsection{Sampling Error}
In the above analysis, the estimation error $\epsilon$ is fixed.
We now explore the relationship between the data distribution and the distribution of $\epsilon$.
Notice that for any random variable $X$, if $X_1, ..., X_m$ are $m$ independent samples, we have
{\small$\mathbb E \left[ \left(\dfrac{1}{m}\sum_{i=1}^m X_i - \mathbb E[X] \right)^2 \right] = \dfrac{1}{m} \mathbb E\left[\left(X -\mathbb EX\right)^2 \right]$}.
In other words, if the random variable is likely to fall far from the mean of its distribution, the sample mean tends to have a larger error.
Although $\epsilon$ is not exactly the error of the sample mean in our case since it is an angle, it is also reasonable to assume $\mathbb E[\epsilon^2]\propto \mathbb E\left[(X - \mathbb EX)^2\right]$.
To make (the magnitude of) $\epsilon$ larger, the data points should be away from their mean as much as possible.
Interestingly, in our case, this is also equivalent to pushing the data points toward the decision boundary.

\subsubsection{Clustering Error}
While the above discussion is on the generalization error under a small number of labeled data, we can easily see that such distribution makes clustering difficult.
This is because the basic idea of clustering is that the data points in the same cluster shall be close to each other, and those in the different clusters shall be far from each other.
While the same-class embeddings are being pushed toward the decision boundary, intra-class distances are increased, while inter-class distances are decreased.
This is exactly opposite to the basic requirement of clustering.
\\

In summary, pushing data points to the boundary of the decision region will increase the learning error of both fine-tuning and clustering for the following reasons:
\begin{itemize}
    \item The sampling error tends to be larger.
    \item A small error in the decision region will result in a large generalization error.
    \item Intra-class distances are larger and inter-class distances are smaller.
\end{itemize}

\begin{figure}
    \centering
    \includegraphics[width=1\linewidth]{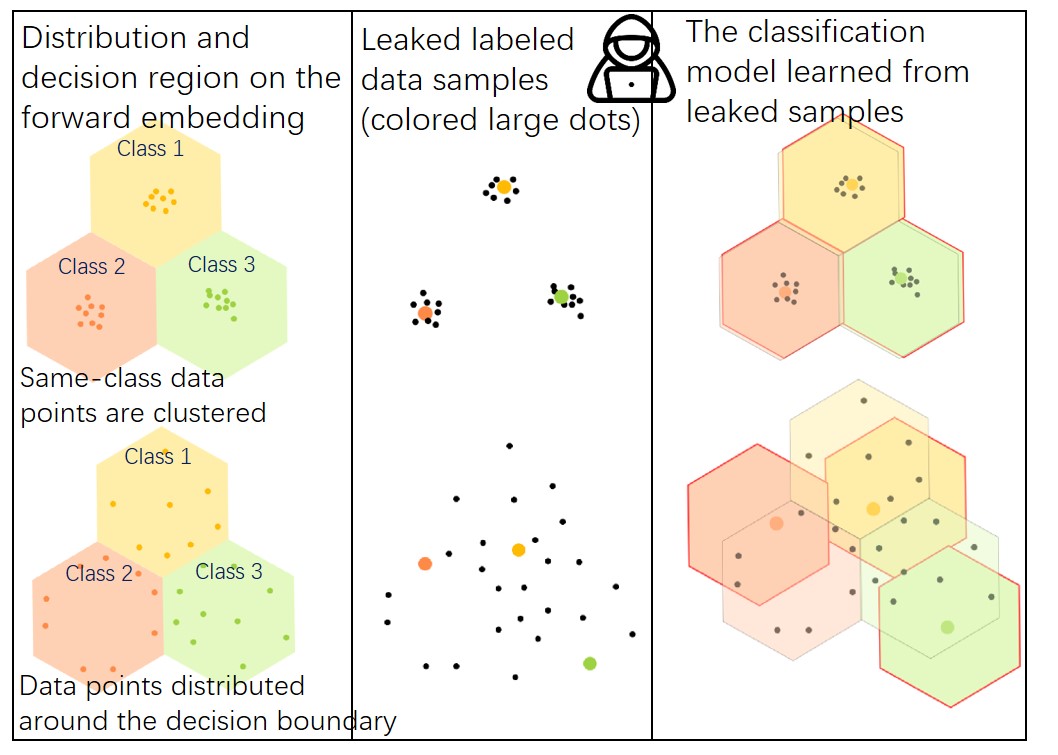}
    \caption{
    A high-level illustration of our idea.
    }
    \label{fig:idea}
\end{figure}

\subsection{Potential Energy Loss}
When electrostatic equilibrium is established, any net charge resides on the surface of the conductor~\cite[ch. 2]{griffiths2005introduction}.
This is partly caused by Coulomb's law, which tells us that \textit{like charges} (electric charges of the same sign) repel each other and \textit{opposite charges} attract each other.
Inspired by this, we can view the embeddings of the same class as like charges and have repulsive forces against each other.
As a result, those data points will tend to be away from each other and be pushed to the boundary of the decision region.

Coulomb's law is stated as follows:
\begin{equation}
\label{eq:coulomb}
    \mathbf F = k {q_1q_2(\mathbf r_1 - \mathbf r_2)}/{\Vert r_1 - r_2 \Vert_2^3},
\end{equation}
where $k$ is a constant, $q_1, q_2$ is the signed magnitude of the charges, $\mathbf r_1, \mathbf r_2$ are their positions, $\Vert\cdot\Vert_2$ is the Euclidean norm, and $\mathbf F$ is the repulsive force on the first charge caused by the second charge.
Since we assume all embeddings belonging to the same class have the same sign and magnitude while ignoring the constant term, \eqref{eq:coulomb} becomes
$\mathbf F = \dfrac{\mathbf r_1 - \mathbf r_2}{\Vert\mathbf r_1 - \mathbf r_2\Vert_2^3}$.
Notice that the repulsive force is the gradient of the electric potential energy, we can further write $\mathbf F$ as $\mathbf F = \nabla_{\mathbf r_1} \dfrac{1}{\Vert \mathbf r_1 - \mathbf r_2 \Vert_2}$,
which is naturally suited to the gradient descent method.
Based on this, we define the \textit{potential energy loss} (PELoss) as
\begin{equation}
\label{eq:loss-pe}
    L_\text{pe} = \sum\limits_{c\in\mathcal C}\sum\limits_{\mathbf z \in Z_c}\sum\limits_{\mathbf z' \in Z_c, \mathbf z' \ne \mathbf z} \dfrac{1}{\Vert \mathbf z - \mathbf z'\Vert_2},
\end{equation}
where $\mathcal C$ is the label set, and $Z_c$ is the forward embeddings of $c$-labeled samples.

By adding $L_\text{pe}$ to the loss function, during the training of the split model, the bottom model outputs of the same class are pushed away from each other, and move towards the boundary of the decision region of its own class.
While in the 3-D case, minimizing the potential energy leads to zero charge density inside the region by Thomson's theorem, it is not necessarily true in the high-dimensional case.
However, we are able to prove a weaker theorem, i.e., minimizing the potential energy leads to a non-zero probability mass in the border of the region.
\begin{restatable}[Border distribution]{theorem}{thmBorderDist}
    Consider a $d$-dimensional bounded region $\Omega \subset \mathbb R^d$, we denote the density of probability distribution in $\Omega$ which minimizes the potential energy functional as 
    \begin{equation}
    \label{eq:potential-min}
        f^* = \mathop\textnormal{argmin}_f \textnormal{PE}(f) = \mathop\textnormal{argmin}_{f} \int_{x \in \Omega} \int_{y \in \Omega} \dfrac{f(x)f(y)}{\Vert x - y \Vert_2} dUdV.
    \end{equation}
    where $f$ is the probability density function of some distribution such that $f \ge 0$ and $\int_{\Omega} f(x)dV = 1$,
    $\Delta_\epsilon \Omega = \{x: x + \epsilon r \not\in \Omega, x \in \Omega, \Vert r \Vert_2 = 1 \}$ is the set of points whose distance to the border of $\Omega$ is less than $\epsilon$.
    Then $f^*$ satisfies that $\int_{x\in\Delta_\epsilon \Omega} f^*(x)dV > 0$ for any $\epsilon > 0 $.
\end{restatable}

\begin{proof}
    See Appendix A.
\end{proof}

\subsubsection{Adding Layer Normalization}
One hidden condition for our method is that the decision boundary is also the actual set boundary of same-class embeddings.
This requires the embedding space to be a borderless manifold.
Otherwise, adding repulsive force may just push the embeddings toward the set boundary instead of the decision boundary.
For example, if the embedding space is the trivial Euclidean space $\mathbb R^d$, the repulsive force will make embeddings far from the origin, while the inter-class distances are still large.
To overcome this, we simply enforce layer normalization~\cite{leiba2016layernorm} (without element-wise affine transformation) on $\mathbf z$, which restricts $\mathbf z$ to the $d$-sphere of radius $\sqrt{d}$, i.e., $\Vert\mathbf z \Vert_2^2=d$.
Accordingly, the distance metric is changed to angular distance, i.e., $\arccos \langle \mathbf z,\mathbf z' \rangle$.
The potential energy loss defined in \eqref{eq:loss-pe} should be changed to 
\begin{equation}
\small
    L_\text{pe} = \sum\limits_{c\in\mathcal C}\sum\limits_{\substack{\mathbf z, \mathbf z' \in Z_c \\ \mathbf z' \ne \mathbf z}}
    \dfrac{1}{\arccos \langle\mathbf z, \mathbf z'\rangle}.
\end{equation}

The combined loss for split training is $L' = L + \alpha L_\text{pe}$, where $L$ is the original loss function (i.e., cross-entropy loss), $\alpha$ is a coefficient to control the intensity of repulsive force.

\subsection{Relationship with Distance Correlation}
Distance correlation is used to protect both the input feature~\cite{Vepakomma2020nopeek} and the label~\cite{sunjiankang2022split_learning_label}.
Here we consider the label-protection case that distance correlation loss is applied to the forward embedding and the label.
The distance correlation loss on one batch is
\begin{equation}
\label{eq:dcor-loss-0}
    L_\text{dcor} = \sum_{i, j=1}^n {d_{i, j}}{d'_{i, j}} \big/ \sqrt{\sum_{i, j=1}^n d^2_{i, j}\sum_{i, j=1}^n d'^2_{i, j}},
\end{equation}
where $d_{i, j}$ is the doubly-centered distance between $i$-th sample's embedding and $j$-th sample's embedding, and $d'_{i,j}$ is the doubly-centered distance between $i$-th label and $j$-th label.
In the classification task, same-class samples have the same label.
If the $i$-th sample and $j$-th sample belong to the same class, we have $d'_{i,j} = 0$.
Thus, \eqref{eq:dcor-loss-0} (ignoring the denominator) is converted to
\begin{equation}
\small
    \sum_{\substack{c, c'\in \mathcal C \\ c\ne c'}}\sum_{\substack{\mathbf z \in Z_c \\ \mathbf z' \in Z_{c'}}} k\left(\Vert\mathbf z - \mathbf z'\Vert_2 - \overline{\Vert\mathbf z - \cdot\Vert_2} - \overline{\Vert\cdot - \mathbf z'\Vert_2} + \overline{\Vert \cdot - \cdot \Vert_2}\right),
\end{equation}
where $\mathcal C$ is the set of labels, $k$ is some constant, $\overline{\Vert\mathbf z - \cdot\Vert_2}$ is the average distance from $\mathbf z$ to other embeddings within the batch (similar for $\overline{\Vert\mathbf \cdot - \mathbf z'\Vert_2}$), and $\overline{\Vert \cdot - \cdot \Vert_2}$ is the average distance between all embedding pairs within the batch.
We can see that minimizing the distance correlation is similar to minimizing the inter-class distances.
As our method is to maximize intra-class distances, minimizing distance correlation has a similar effect.
However, in an intuitive understanding, the fact that embeddings of different classes lie in different directions makes distance correlation naturally noisy.
Experiments in \Cref{sec:exp-pair} also illustrate that minimizing distance correlation fails to push away some same-class samples.

\section{Empirical Study}
\begin{figure*}[t]
    \centering
    \includegraphics[width=1\linewidth]{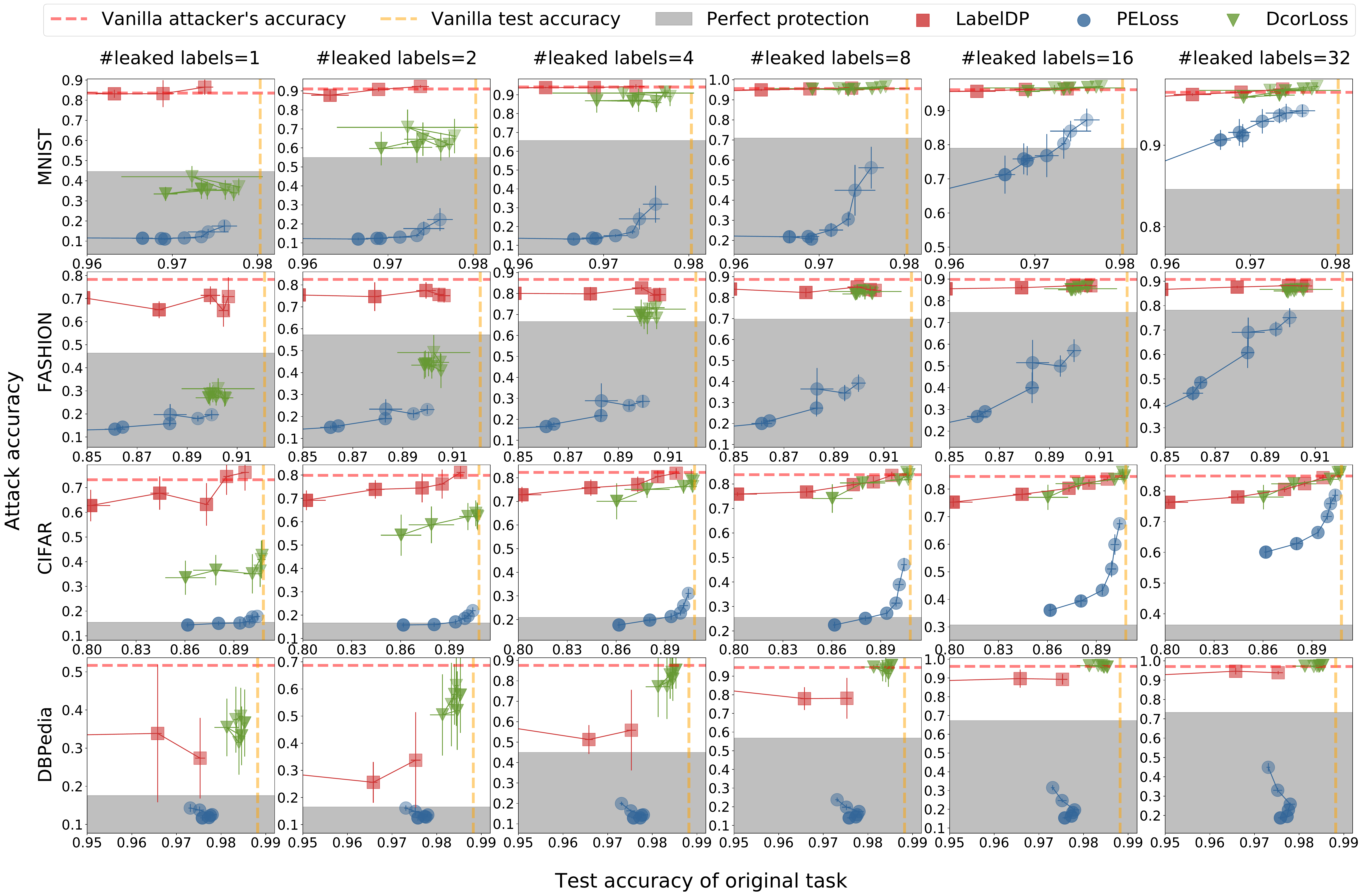}
    \caption{Test accuracy vs. attack accuracy of the fine-tuning attack. 
        The data point lies at the lower-right position means the experiment achieves higher test accuracy on the original task with lower attacker's accuracy, which is our desired result.
        The data point lies in the gray area means that the attacker's accuracy is lower than training from scratch without the bottom model's information, i.e., perfect protection.
        }
    \label{fig:fine-tuning_attack}
\end{figure*}

\begin{figure*}[h]
    \centering
    \includegraphics[width=1\linewidth]{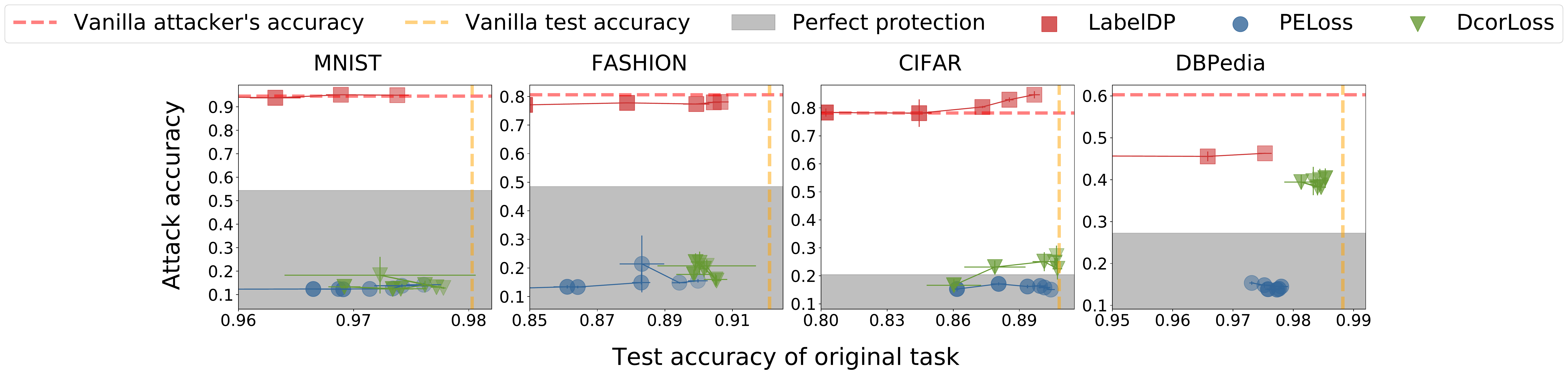}
    \caption{Test accuracy vs. attack accuracy of the clustering attack.
    The notations are the same as \Cref{fig:fine-tuning_attack}.}
    \label{fig:clustering_attack}
\end{figure*}

We conduct experiments on four different datasets, i.e., MNIST~\cite{mnist}, Fashion-MNIST~\cite{fashion}, CIFAR-10~\cite{cifar}, and DBpedia~\cite{2007dbpedia}.
We compare the attack performance of vanilla split training, training with potential energy loss (PELoss, our method), training with distance correlation loss (DcorLoss) proposed by \cite{Vepakomma2020nopeek,sunjiankang2022split_learning_label}, and a simple baseline of label differential privacy method based on randomly flipping a certain portion of labels (LabelDP).
For DcorLoss, layer normalization is also added like in our method for training stability.
The attacks include both fine-tuning attacks and clustering attacks.
For clustering attacks, we use the classical k-Means algorithm~\cite{macqueen1967kmeans}, with the number of classes known to the attacker.
We also measure the distances between sample pairs to illustrate the mechanism of PELoss.

Due to space limits, detailed descriptions of the model architectures and training strategies, and more results are presented in Appendix B and C.
Other results include studies on different split positions, different attack layers, different forward embedding dimensions, and the t-SNE~\cite{van2008tsne} visualization of the forward embeddings.

\subsection{Experiment Settings}
We implement the experiment codes using the PyTorch and Scikit-Learn~\cite{scikit-learn} libraries,
and run them on servers with NVIDIA RTX3090 GPUs.
Each experiment is repeated 3 times for the original tasks and 5 times for the attack tasks, with different random seeds.
We use a 3-layer fully-connected network for MNIST, a simple convolutional network for Fashion-MNIST, ResNet-20~\cite{hekaiming2016resnet} for CIFAR-10 dataset, and TextCNN~\cite{kimyoon2014textcnn} for DBPedia.
The split position for each model is the last dense layer by default since the last forward embedding is the closest to the label and is the most difficult to protect.
For the selection of hyperparameters, we vary the loss coefficient ($\alpha$)  of PELoss from $0.25\sim 32$ and $1\sim 32$ for DcorLoss.
For LabelDP, the ratio of randomly flipped labels varies from $0.01\sim 0.16$.
In all experiments, the value doubles each time.
Detailed experiment settings are provided in Appendix B.

\subsection{Fine-tuning Attack}
We report the test accuracy of the original task and the accuracy of fine-tuning attacks using different protection methods in \Cref{fig:fine-tuning_attack}.
Deeper colors represent stronger protection, e.g., larger loss coefficient or label flipping probability, while lighter colors mean more leaning toward preserving the test accuracy, e.g., smaller loss coefficient or label flipping probability.
We also plot the test accuracy (orange dashed line) and the attack accuracy (red dashed line) in the vanilla split learning case, and the perfect protection area (where the attack accuracy is lower than the accuracy of training from scratch using leaked labels).
For all methods, we observe that decreasing the attack accuracy usually also lowers test accuracy on the original task,
and more leaked labels lead to higher attack accuracy.
Although all methods protect privacy to some extent at the cost of damaging the model performance, it is obvious that our proposed PELoss is superior to DcorLoss and LabelDP.
PELoss has the following advantages compared with other methods:
\begin{itemize}
    \item The curves of PELoss are constantly at the lower-right side of other methods.
    In other words, on the same test accuracy level, the PELoss has a significantly lower attacker's accuracy than DcorLoss and LabelDP.
    
    \item The curves are smoother and the error bars are also smaller using PELoss, indicating it is more responsive to the change of $\alpha$ and the performance is more stable.
    Thus, it is easier to balance privacy and model performance using PELoss.
    \item PELoss has many data points in the perfect protection zone, while other methods rarely achieve perfect protection.
\end{itemize}

In summary, PELoss provides significantly stronger privacy protection against model completion attacks, while preserving the model performance better than other methods.
We also notice that in some datasets, varying the coefficients of DcorLoss seems to have no effect, however, when the coefficient is large enough (e.g., 32), the training could diverge.

\subsection{Clustering Attack}
We report the test accuracy of the original task and the accuracy of clustering attacks using different protection methods in \Cref{fig:clustering_attack}.
In all the experiments, the attacker can access the entire unlabeled test dataset which contains more than 10,000 unlabeled samples, and the number of classes is known to the attacker.
The notations are the same as fine-tuning attacks, while here the gray area (perfect protection) means that the accuracy of clustering on forward embeddings is lower than directly clustering on the raw input samples.
The accuracy here is defined as the maximum accuracy among all possible cluster-label assignments.

We can see that PELoss achieves perfect protection in all cases (i.e., the clustering result on the forward embeddings is worse than on the raw input data), and is significantly better than DcorLoss on CIFAR and DBPedia.
On the other side, LabelDP performs badly against clustering attacks.

\subsection{Distances between Sample Pairs}
\label{sec:exp-pair}
\begin{figure}[h]
    \centering
    \includegraphics[width=1\linewidth]{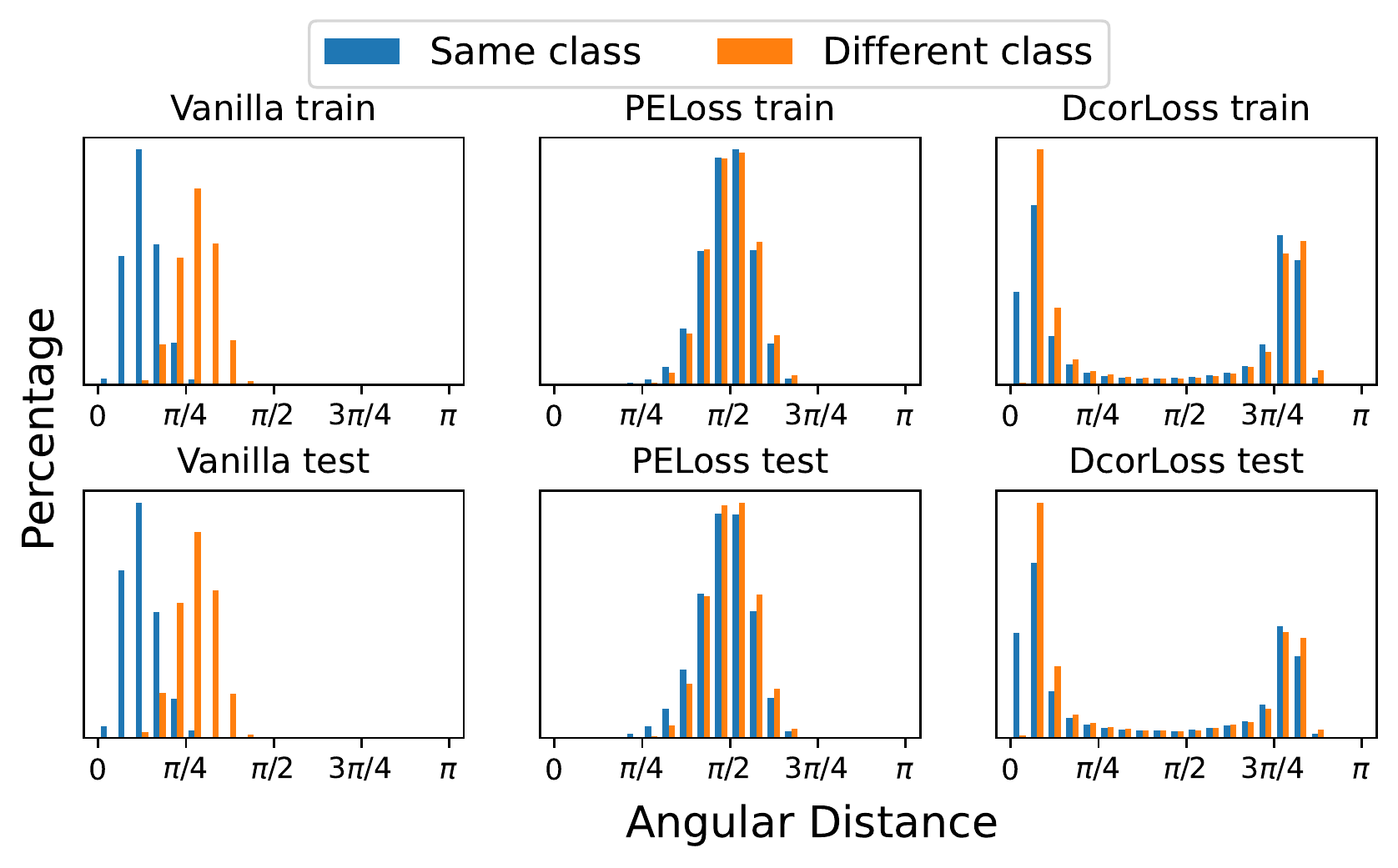}
    \caption{Angular distance distribution between sample pairs on MNIST. $\alpha=1$ for PELoss and 0.5 for DcorLoss, as they have similar performance on the original task.}
    \label{fig:angular-dist}
\end{figure}
To better illustrate the distribution of forward embeddings, we plot the distribution of angular distances between same-class and different-class sample pairs in \Cref{fig:angular-dist}.
For vanilla training, the angular distances are small between same-class samples, and are large between different-class samples.
For PELoss, the angular distances between sample pairs are close to $\pi/2$, no matter whether they belong to the same class or not, in both training data and test data.
In contrast, for DcorLoss, the angular distances appear to have a double-peak distribution.
We can see that there are significantly more same-class sample pairs that have an angular distance around 0, which weakens the protection.

\section{Conclusion}
In this paper, we investigate the privacy leakage of split learning arising from the learned bottom model.
We view the model completion attack as a learning procedure for the attacker and turn the privacy-preserving problem into a learning problem.
We find that pushing embeddings to their decision boundary increases the learning error for the attacker, and propose the potential energy loss on forward embedding to protect the privacy.
Extensive experiments show that our method significantly restricts the ability of the attacker's fine-tuning and clustering attacks while reserving the model performance, superior to baselines in terms of both utility and privacy.
The limitation of this work mainly includes that only the inference process is protected, and the lack of a theoretical leakage bound.
%

\section*{Acknowledgments}
This work is supported in part by the National Natural Science Foundation of China (No. 72192823), the “Ten Thousand Talents Program” of Zhejiang Province for Leading Experts (No. 2021R52001), and Ant Group.

\bibliographystyle{named}
\bibliography{ref}

\begin{thebibliography}{}

\bibitem[\protect\citeauthoryear{Abuadbba \bgroup \em et al.\egroup
  }{2020}]{abuadbba2020split_cnn}
Sharif Abuadbba, Kyuyeon Kim, Minki Kim, Chandra Thapa, Seyit~Ahmet
  {\c{C}}amtepe, Yansong Gao, Hyoungshick Kim, and Surya Nepal.
\newblock Can we use split learning on 1d {CNN} models for privacy preserving
  training?
\newblock In {\em {ASIA} {CCS} '20: The 15th {ACM} Asia Conference on Computer
  and Communications Security}, pages 305--318. {ACM}, 2020.

\bibitem[\protect\citeauthoryear{Auer \bgroup \em et al.\egroup
  }{2007}]{2007dbpedia}
S{\"{o}}ren Auer, Christian Bizer, Georgi Kobilarov, Jens Lehmann, Richard
  Cyganiak, and Zachary~G. Ives.
\newblock Dbpedia: {A} nucleus for a web of open data.
\newblock In {\em The Semantic Web, 6th International Semantic Web Conference,
  2nd Asian Semantic Web Conference, {ISWC} 2007 + {ASWC} 2007}, volume 4825 of
  {\em Lecture Notes in Computer Science}, pages 722--735. Springer, 2007.

\bibitem[\protect\citeauthoryear{Ba \bgroup \em et al.\egroup
  }{2016}]{leiba2016layernorm}
Jimmy~Lei Ba, Jamie~Ryan Kiros, and Geoffrey~E. Hinton.
\newblock Layer normalization, 2016.

\bibitem[\protect\citeauthoryear{Berthelot \bgroup \em et al.\egroup
  }{2019}]{berthelot2019mixmatch}
David Berthelot, Nicholas Carlini, Ian~J. Goodfellow, Nicolas Papernot, Avital
  Oliver, and Colin Raffel.
\newblock Mixmatch: {A} holistic approach to semi-supervised learning.
\newblock In {\em Advances in Neural Information Processing Systems 32, NeurIPS
  2019}, pages 5050--5060, 2019.

\bibitem[\protect\citeauthoryear{Cantareira \bgroup \em et al.\egroup
  }{2020}]{cantareira2020hidden}
Gabriel~Dias Cantareira, Elham Etemad, and Fernando~V. Paulovich.
\newblock Exploring neural network hidden layer activity using vector fields.
\newblock {\em Inf.}, 11(9):426, 2020.

\bibitem[\protect\citeauthoryear{Chen \bgroup \em et al.\egroup
  }{2022a}]{ccc2022ppcdr}
Chaochao Chen, Huiwen Wu, Jiajie Su, Lingjuan Lyu, Xiaolin Zheng, and Li~Wang.
\newblock Differential private knowledge transfer for privacy-preserving
  cross-domain recommendation.
\newblock In {\em {WWW} '22: The {ACM} Web Conference 2022, Virtual Event,
  Lyon, France, April 25 - 29, 2022}, pages 1455--1465. {ACM}, 2022.

\bibitem[\protect\citeauthoryear{Chen \bgroup \em et al.\egroup
  }{2022b}]{ccc2022vfgnn}
Chaochao Chen, Jun Zhou, Longfei Zheng, Huiwen Wu, Lingjuan Lyu, Jia Wu,
  Bingzhe Wu, Ziqi Liu, Li~Wang, and Xiaolin Zheng.
\newblock Vertically federated graph neural network for privacy-preserving node
  classification.
\newblock In {\em Proceedings of the Thirty-First International Joint
  Conference on Artificial Intelligence, {IJCAI} 2022}, pages 1959--1965, 2022.

\bibitem[\protect\citeauthoryear{Duan \bgroup \em et al.\egroup
  }{2023}]{duanlin2023privascissors}
Lin Duan, Jingwei Sun, Yiran Chen, and Maria Gorlatova.
\newblock Privascissors: Enhance the privacy of collaborative inference through
  the lens of mutual information.
\newblock {\em CoRR}, abs/2306.07973, 2023.

\bibitem[\protect\citeauthoryear{Erdo{\u{g}}an \bgroup \em et al.\egroup
  }{2022}]{erdougan2022unsplit}
Ege Erdo{\u{g}}an, Alptekin K{\"u}p{\c{c}}{\"u}, and A~Erc{\"u}ment
  {\c{C}}i{\c{c}}ek.
\newblock Unsplit: Data-oblivious model inversion, model stealing, and label
  inference attacks against split learning.
\newblock In {\em Proceedings of the 21st Workshop on Privacy in the Electronic
  Society}, pages 115--124, 2022.

\bibitem[\protect\citeauthoryear{Fu \bgroup \em et al.\egroup
  }{2022a}]{fu2022vertical_federated_label}
Chong Fu, Xuhong Zhang, Shouling Ji, Jinyin Chen, Jingzheng Wu, Shanqing Guo,
  Jun Zhou, Alex~X Liu, and Ting Wang.
\newblock Label inference attacks against vertical federated learning.
\newblock In {\em 31st USENIX Security Symposium (USENIX Security 22)}, 2022.

\bibitem[\protect\citeauthoryear{Fu \bgroup \em et al.\egroup
  }{2022b}]{fufangcheng2022blindFL}
Fangcheng Fu, Huanran Xue, Yong Cheng, Yangyu Tao, and Bin Cui.
\newblock Blindfl: Vertical federated machine learning without peeking into
  your data.
\newblock In {\em {SIGMOD} '22: International Conference on Management of
  Data}, pages 1316--1330. {ACM}, 2022.

\bibitem[\protect\citeauthoryear{Griffiths}{2005}]{griffiths2005introduction}
David~J Griffiths.
\newblock Introduction to electrodynamics, 2005.

\bibitem[\protect\citeauthoryear{Gupta and
  Raskar}{2018}]{gupta2018distributed_learning}
Otkrist Gupta and Ramesh Raskar.
\newblock Distributed learning of deep neural network over multiple agents.
\newblock {\em J. Netw. Comput. Appl.}, 116:1--8, 2018.

\bibitem[\protect\citeauthoryear{Ha \bgroup \em et al.\egroup
  }{2021}]{jeong2021split_medical}
Yoo~Jeong Ha, Minjae Yoo, Gusang Lee, Soyi Jung, Sae~Won Choi, Joongheon Kim,
  and Seehwan Yoo.
\newblock Spatio-temporal split learning for privacy-preserving medical
  platforms: Case studies with covid-19 ct, x-ray, and cholesterol data.
\newblock {\em IEEE Access}, 9:121046--121059, 2021.

\bibitem[\protect\citeauthoryear{He \bgroup \em et al.\egroup
  }{2016}]{hekaiming2016resnet}
Kaiming He, Xiangyu Zhang, Shaoqing Ren, and Jian Sun.
\newblock Deep residual learning for image recognition.
\newblock In {\em 2016 {IEEE} Conference on Computer Vision and Pattern
  Recognition, {CVPR} 2016}, pages 770--778. {IEEE} Computer Society, 2016.

\bibitem[\protect\citeauthoryear{Jin \bgroup \em et al.\egroup
  }{2020}]{jinpengzhan2020generalization}
Pengzhan Jin, Lu~Lu, Yifa Tang, and George~Em Karniadakis.
\newblock Quantifying the generalization error in deep learning in terms of
  data distribution and neural network smoothness.
\newblock {\em Neural Networks}, 130:85--99, 2020.

\bibitem[\protect\citeauthoryear{Kariyappa and
  Qureshi}{2023}]{sanjay2023exploit_split_learning}
Sanjay Kariyappa and Moinuddin~K Qureshi.
\newblock Exploit: Extracting private labels in split learning.
\newblock In {\em 2023 IEEE Conference on Secure and Trustworthy Machine
  Learning (SaTML)}, pages 165--175, 2023.

\bibitem[\protect\citeauthoryear{Kim}{2014}]{kimyoon2014textcnn}
Yoon Kim.
\newblock Convolutional neural networks for sentence classification.
\newblock In {\em Proceedings of the 2014 Conference on Empirical Methods in
  Natural Language Processing, {EMNLP} 2014}, pages 1746--1751. {ACL}, 2014.

\bibitem[\protect\citeauthoryear{Koda \bgroup \em et al.\egroup
  }{2020}]{yusuke2020split_power}
Yusuke Koda, Jihong Park, Mehdi Bennis, Koji Yamamoto, Takayuki Nishio,
  Masahiro Morikura, and Kota Nakashima.
\newblock Communication-efficient multimodal split learning for mmwave received
  power prediction.
\newblock {\em IEEE Communications Letters}, 24(6):1284--1288, 2020.

\bibitem[\protect\citeauthoryear{Koltchinskii and
  Panchenko}{2002}]{koltchinskii2002empirical}
Vladimir Koltchinskii and Dmitry Panchenko.
\newblock Empirical margin distributions and bounding the generalization error
  of combined classifiers.
\newblock {\em The Annals of Statistics}, 30(1):1--50, 2002.

\bibitem[\protect\citeauthoryear{Kontorovich and
  Weiss}{2014}]{kontorovich2014knn_generalization}
Aryeh Kontorovich and Roi Weiss.
\newblock Maximum margin multiclass nearest neighbors.
\newblock In {\em Proceedings of the 31th International Conference on Machine
  Learning, {ICML} 2014}, volume~32 of {\em {JMLR} Workshop and Conference
  Proceedings}, pages 892--900, 2014.

\bibitem[\protect\citeauthoryear{Krizhevsky \bgroup \em et al.\egroup
  }{2009}]{cifar}
Alex Krizhevsky, Geoffrey Hinton, et~al.
\newblock Learning multiple layers of features from tiny images.
\newblock 2009.

\bibitem[\protect\citeauthoryear{LeCun \bgroup \em et al.\egroup
  }{1998}]{mnist}
Yann LeCun, L{\'{e}}on Bottou, Yoshua Bengio, and Patrick Haffner.
\newblock Gradient-based learning applied to document recognition.
\newblock {\em Proc. {IEEE}}, 86(11):2278--2324, 1998.

\bibitem[\protect\citeauthoryear{Lei \bgroup \em et al.\egroup
  }{2019}]{leiyunwen2019data-dependent}
Yunwen Lei, {\"{U}}r{\"{u}}n Dogan, Ding{-}Xuan Zhou, and Marius Kloft.
\newblock Data-dependent generalization bounds for multi-class classification.
\newblock {\em {IEEE} Trans. Inf. Theory}, 65(5):2995--3021, 2019.

\bibitem[\protect\citeauthoryear{Li \bgroup \em et al.\egroup
  }{2022}]{li2021split_learning_label}
Oscar Li, Jiankai Sun, Xin Yang, Weihao Gao, Hongyi Zhang, Junyuan Xie,
  Virginia Smith, and Chong Wang.
\newblock Label leakage and protection in two-party split learning.
\newblock In {\em The Tenth International Conference on Learning
  Representations, {ICLR} 2022}, 2022.

\bibitem[\protect\citeauthoryear{Luo \bgroup \em et al.\egroup
  }{2021}]{luo2021inference_attack}
Xinjian Luo, Yuncheng Wu, Xiaokui Xiao, and Beng~Chin Ooi.
\newblock Feature inference attack on model predictions in vertical federated
  learning.
\newblock In {\em 37th {IEEE} International Conference on Data Engineering,
  {ICDE} 2021}, pages 181--192. {IEEE}, 2021.

\bibitem[\protect\citeauthoryear{MacQueen and
  others}{1967}]{macqueen1967kmeans}
James MacQueen et~al.
\newblock Some methods for classification and analysis of multivariate
  observations.
\newblock In {\em Proceedings of the fifth Berkeley symposium on mathematical
  statistics and probability}, volume~1, pages 281--297, 1967.

\bibitem[\protect\citeauthoryear{Mohassel and
  Zhang}{2017}]{mohassel2017secureml}
Payman Mohassel and Yupeng Zhang.
\newblock Secureml: {A} system for scalable privacy-preserving machine
  learning.
\newblock In {\em 2017 {IEEE} Symposium on Security and Privacy, {SP} 2017},
  pages 19--38. {IEEE} Computer Society, 2017.

\bibitem[\protect\citeauthoryear{Negrea \bgroup \em et al.\egroup
  }{2019}]{negrea2019information_generalization}
Jeffrey Negrea, Mahdi Haghifam, Gintare~Karolina Dziugaite, Ashish Khisti, and
  Daniel~M. Roy.
\newblock Information-theoretic generalization bounds for {SGLD} via
  data-dependent estimates.
\newblock In {\em Advances in Neural Information Processing Systems 32, NeurIPS
  2019}, pages 11013--11023, 2019.

\bibitem[\protect\citeauthoryear{Pasquini \bgroup \em et al.\egroup
  }{2021}]{pasquini2021inference_attack}
Dario Pasquini, Giuseppe Ateniese, and Massimo Bernaschi.
\newblock Unleashing the tiger: Inference attacks on split learning.
\newblock In {\em {CCS} '21: 2021 {ACM} {SIGSAC} Conference on Computer and
  Communications Security}, pages 2113--2129. {ACM}, 2021.

\bibitem[\protect\citeauthoryear{Pedregosa \bgroup \em et al.\egroup
  }{2011}]{scikit-learn}
Fabian Pedregosa, Ga{\"{e}}l Varoquaux, Alexandre Gramfort, Vincent Michel,
  Bertrand Thirion, Olivier Grisel, Mathieu Blondel, Peter Prettenhofer, Ron
  Weiss, Vincent Dubourg, Jake VanderPlas, Alexandre Passos, David Cournapeau,
  Matthieu Brucher, Matthieu Perrot, and Edouard Duchesnay.
\newblock Scikit-learn: Machine learning in python.
\newblock {\em J. Mach. Learn. Res.}, 12:2825--2830, 2011.

\bibitem[\protect\citeauthoryear{Pensia \bgroup \em et al.\egroup
  }{2018}]{pensia2018generalization_iterative}
Ankit Pensia, Varun~S. Jog, and Po{-}Ling Loh.
\newblock Generalization error bounds for noisy, iterative algorithms.
\newblock In {\em 2018 {IEEE} International Symposium on Information Theory,
  {ISIT} 2018, Vail, CO, USA, June 17-22, 2018}, pages 546--550. {IEEE}, 2018.

\bibitem[\protect\citeauthoryear{Pezzotti \bgroup \em et al.\egroup
  }{2018}]{pezzotti2018deepeyes}
Nicola Pezzotti, Thomas H{\"{o}}llt, Jan~C. van Gemert, Boudewijn P.~F.
  Lelieveldt, Elmar Eisemann, and Anna Vilanova.
\newblock Deepeyes: Progressive visual analytics for designing deep neural
  networks.
\newblock {\em {IEEE} Trans. Vis. Comput. Graph.}, 24(1):98--108, 2018.

\bibitem[\protect\citeauthoryear{Rathee \bgroup \em et al.\egroup
  }{2020}]{rathee2020cryptflow2}
Deevashwer Rathee, Mayank Rathee, Nishant Kumar, Nishanth Chandran, Divya
  Gupta, Aseem Rastogi, and Rahul Sharma.
\newblock Cryptflow2: Practical 2-party secure inference.
\newblock In {\em {CCS} '20: 2020 {ACM} {SIGSAC} Conference on Computer and
  Communications Securit}, pages 325--342. {ACM}, 2020.

\bibitem[\protect\citeauthoryear{Rauber \bgroup \em et al.\egroup
  }{2017}]{rauber2017visualize_hidden}
Paulo~E. Rauber, Samuel~G. Fadel, Alexandre~X. Falc{\~{a}}o, and Alexandru~C.
  Telea.
\newblock Visualizing the hidden activity of artificial neural networks.
\newblock {\em {IEEE} Trans. Vis. Comput. Graph.}, 23(1):101--110, 2017.

\bibitem[\protect\citeauthoryear{Russo and Zou}{2020}]{russo2020overfit}
Daniel Russo and James Zou.
\newblock How much does your data exploration overfit? controlling bias via
  information usage.
\newblock {\em {IEEE} Trans. Inf. Theory}, 66(1):302--323, 2020.

\bibitem[\protect\citeauthoryear{Sun \bgroup \em et al.\egroup
  }{2022}]{sunjiankang2022split_learning_label}
Jiankai Sun, Xin Yang, Yuanshun Yao, and Chong Wang.
\newblock Label leakage and protection from forward embedding in vertical
  federated learning.
\newblock {\em CoRR}, abs/2203.01451, 2022.

\bibitem[\protect\citeauthoryear{Székely \bgroup \em et al.\egroup
  }{2007}]{szekeley2007dcor}
Gábor~J. Székely, Maria~L. Rizzo, and Nail~K. Bakirov.
\newblock {Measuring and testing dependence by correlation of distances}.
\newblock {\em The Annals of Statistics}, 35(6):2769 -- 2794, 2007.

\bibitem[\protect\citeauthoryear{Van~der Maaten and Hinton}{2008}]{van2008tsne}
Laurens Van~der Maaten and Geoffrey Hinton.
\newblock Visualizing data using t-sne.
\newblock {\em Journal of machine learning research}, 9(11), 2008.

\bibitem[\protect\citeauthoryear{Vepakomma \bgroup \em et al.\egroup
  }{2018}]{vepakomma2018split_health}
Praneeth Vepakomma, Otkrist Gupta, Tristan Swedish, and Ramesh Raskar.
\newblock Split learning for health: Distributed deep learning without sharing
  raw patient data, 2018.

\bibitem[\protect\citeauthoryear{Vepakomma \bgroup \em et al.\egroup
  }{2020}]{Vepakomma2020nopeek}
Praneeth Vepakomma, Abhishek Singh, Otkrist Gupta, and Ramesh Raskar.
\newblock Nopeek: Information leakage reduction to share activations in
  distributed deep learning.
\newblock In {\em 20th International Conference on Data Mining Workshops,
  {ICDM} Workshops 2020}, pages 933--942. {IEEE}, 2020.

\bibitem[\protect\citeauthoryear{Xiao \bgroup \em et al.\egroup
  }{2017}]{fashion}
Han Xiao, Kashif Rasul, and Roland Vollgraf.
\newblock Fashion-mnist: a novel image dataset for benchmarking machine
  learning algorithms, 2017.

\bibitem[\protect\citeauthoryear{Xu \bgroup \em et al.\egroup
  }{2021}]{xuyi2021dp_ssl}
Yi~Xu, Jiandong Ding, Lu~Zhang, and Shuigeng Zhou.
\newblock {DP-SSL:} towards robust semi-supervised learning with {A} few
  labeled samples.
\newblock In {\em Advances in Neural Information Processing Systems 34, NeurIPS
  2021}, pages 15895--15907, 2021.

\bibitem[\protect\citeauthoryear{Zhicong \bgroup \em et al.\egroup
  }{2022}]{huangzhicong2022cheetah}
Huang Zhicong, Lu~Wen-jie, Hong Cheng, and Ding Jiansheng.
\newblock Cheetah: Lean and fast secure {Two-Party} deep neural network
  inference.
\newblock In {\em 31st USENIX Security Symposium (USENIX Security 22)}. USENIX
  Association, August 2022.

\bibitem[\protect\citeauthoryear{Zhou \bgroup \em et al.\egroup
  }{2022}]{chenchaochao2020codesign}
Jun Zhou, Longfei Zheng, Chaochao Chen, Yan Wang, Xiaolin Zheng, Bingzhe Wu,
  Cen Chen, Li~Wang, and Jianwei Yin.
\newblock Toward scalable and privacy-preserving deep neural network via
  algorithmic-cryptographic co-design.
\newblock {\em ACM Transactions on Intelligent Systems and Technology (TIST)},
  13(4):1--21, 2022.

\end{thebibliography}

\clearpage

\appendix
\begin{table*}[t]
\centering
\begin{tabular}{@{}cc@{}}
\toprule
Task        & Network Structure                                                                                                                                                                                  \\ \midrule
MNIST       & Linear(784, 128)-LeakyReLU-Linear(128, 32)-LeakyReLU-Linear(32, 10)                                                                                                                                \\ \midrule
FASHION     & \begin{tabular}[c]{@{}c@{}}Conv(1,32,kernel\_size=5)-LeakyReLU-MaxPool(2)-\\ Conv(32, 64, kernel\_size=3, padding=1)-LeakyReLU-MaxPool(2)-\\ Linear(2304, 128)-Tanh-Linear(128, 10)\end{tabular} \\ \midrule
CIFAR       & ResNet-20 described in the original paper.                                                                                                                                                          \\ \midrule
DBPedia     & \begin{tabular}[c]{@{}c@{}}TextCNN with filter sizes {[}3, 4, 5{]} and 200 channels, \\ using glove.6B.50d pretrained word embeddings.\end{tabular}                                              \\ \bottomrule
\end{tabular}
\caption{Model architectures used in the experiments.}
\label{tab:model-arch}
\end{table*}

\begin{figure*}[h]
    \centering
    \includegraphics[width=1\linewidth]{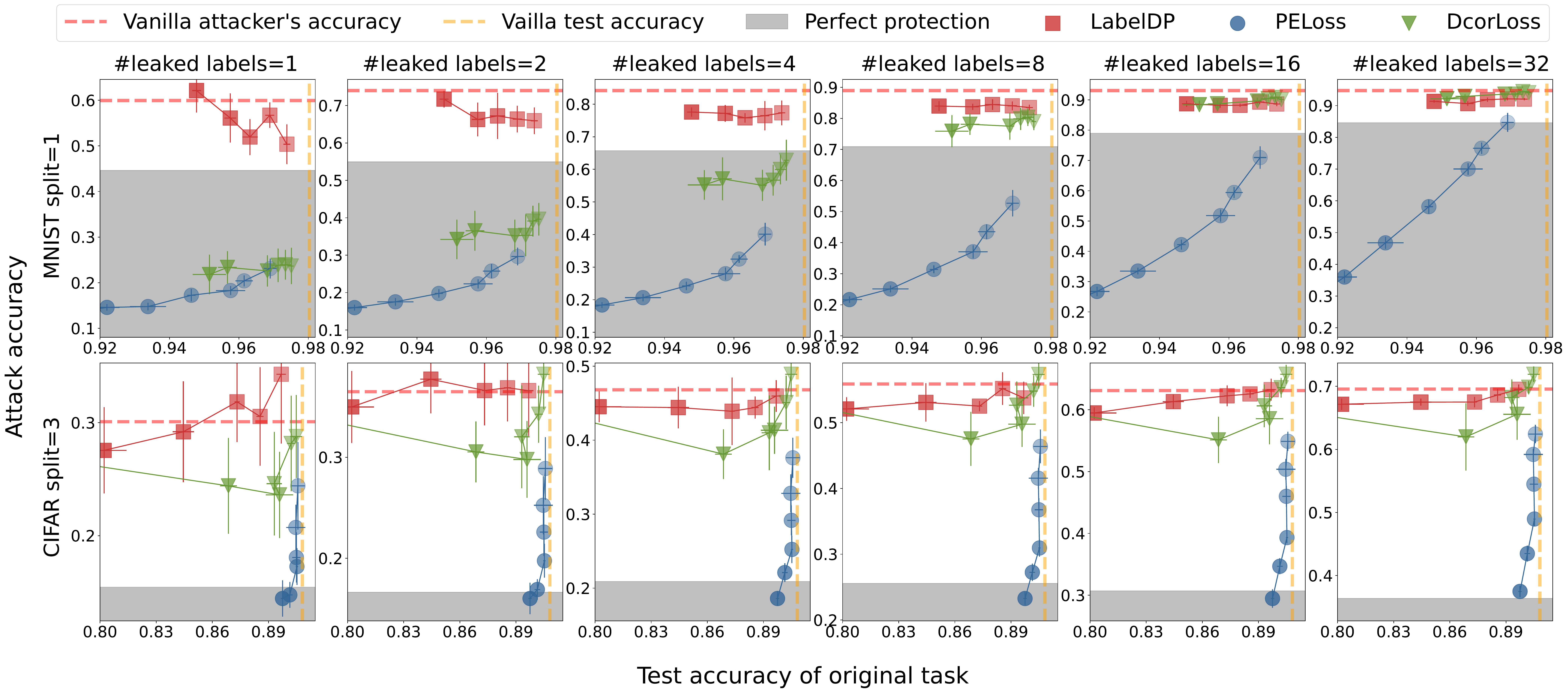}
    \caption{Test accuracy of original task vs. accuracy of fine-tuning attack.}
    \label{fig:fine-tuning-middle}
\end{figure*}

\section{Proof of Theorem 1}
\thmBorderDist*
\begin{proof}
Note we use v instead of $\Vert\cdot\Vert_2$ to denote the Euclidean norm here for simplicity.
Consider the ball $B_{r}(x_0)$ centered at $x_0$ with radius of $r$.
Let 
\begin{equation}
    r_0 = \inf_{|r|} \left\{ \int_{y\in \Omega - B_{r}(x_0)} f(y) dV = 0 \right\},
\end{equation}
i.e., $f(x)$ is zero almost everywhere outside of $B_{r_0}(x_0)$.
Since $B_{r_0}(x_0)$ is the smallest ball originated at $x_0$ that contains all mass of $f(x)$, for some small number $\delta > 0$, we can find a small ball $B_{\epsilon/4}(x + r')$ near the surface of $B_{r_0}(x_0)$ such that 
\begin{equation}
    \int_{y\in B_{\epsilon/4}(x + r')} f(y)dV = \delta,
\end{equation}
where $|r'| = r_0 - \dfrac{\epsilon}{4}$.

Now consider moving the probability mass of $f$ in $B_{\epsilon/4}(x + r')$ to $B_{\epsilon/4}(x + r'')$, where $B_{\epsilon/4}(r'') \subset \Delta_\epsilon \Omega$ and $|r''| \ge r_0 + \epsilon$.
It is feasible because $r' + 1/4\epsilon e_r \in \Omega - \Delta_\epsilon \Omega$ and $B_{\epsilon}(r' + 3/4\epsilon e_r) \subset \Omega$.
Here we use $e_r$ to denote the unit vector with the same direction as $r'$.
In other words, we can have the following new probability distribution
\begin{equation}
    f'(t) = \begin{cases}
        f(t - r'' + r') & \quad \text{for $t \in B_{\epsilon/4}(x + r'')$},\\
        0               & \quad \text{for $t \in B_{\epsilon/4}(x + r')$},\\
        f(t)            & \quad \text{otherwise}.
    \end{cases}
\end{equation}
For convenience, we denote $B_{\epsilon/4}(x + r')$ as $B_1$, $B_{\epsilon/4}(x + r'')$ as $B_2$, $\Omega - B_1 - B_2$ as $\Omega'$.
The potential energy difference between $f$ and $f'$ can be calculated as follows:

\begin{align}
\small
    &\text{PE}(f) - \text{PE}(f') = \int_{x \in \Omega} \int_{y \in \Omega} \dfrac{f(x)f(y) - f'(x)f'(y)}{|x-y|}dUdV 
    \nonumber \\
    & = 
    \int_{x \in B_1 \cup B_2} \int_{y \in \Omega'} \dfrac{f(x)f(y) - f'(x)f'(y)}{|x - y|}dUdV 
    \nonumber\\ & \quad +
     \int_{x \in \Omega'} \int_{y \in B_1 \cup B_2} \dfrac{f(x)f(y) - f'(x)f'(y)}{|x-y|}dUdV
    \nonumber\\ & \quad +
     \int_{x \in B_1 \cup B_2} \int_{y \in B_1 \cup B_2} \dfrac{f(x)f(y) - f'(x)f'(y)}{|x-y|}dUdV 
    \displaybreak\nonumber\\     
     & = 2 \int_{x \in B_1} \int_{y \in \Omega'} \dfrac{f(x)f(y)}{|x-y|}dUdV
     \nonumber\\ & \quad -
     2 \int_{x \in B_2} \int_{y \in \Omega'} \dfrac{f'(x)f(y)}{|x-y|}dUdV
     \nonumber\\ & \quad +
     \int_{x \in B_1} \int_{y \in B_1} \dfrac{f(x)f(y)}{|x-y|}dUdV
     \nonumber\\ & \quad -
     \int_{x \in B_2} \int_{y \in B_2} \dfrac{f'(x)f'(y)}{|x-y|}dUdV.
\end{align}

Notice that $f'$ is just moving the probability mass of $f$ from $B_1(r')$ to $B_2$, the latter two terms are canceled out.
Thus, we only need to consider the term 
\begin{equation}
\label{eq:pe-diff-2}
\begin{split}
    &\dfrac{1}{2}[\text{PE}(f) - \text{PE}(f')]= \\ 
    &\int_{x \in B_1} \int_{y \in \Omega'} \dfrac{f(x)f(y)}{|x-y|}dUdV
    \\ &-
    \int_{x \in B_2} \int_{y \in \Omega'} \dfrac{f'(x)f(y)}{|x-y|}dUdV.
\end{split}
\end{equation}
We can define the probability in the ball to be
\begin{equation}
    g(s) = f(x - r') = f'(x - r''),
\end{equation}
where $s \in B_{\epsilon / 4}(0)$.
Thus, \eqref{eq:pe-diff-2} can be further written as
\begin{equation}
\label{eq:pe-diff-3}
\small
    \int_{x \in B_{\epsilon/4}(0)}\int_{y \in \Omega'} 
    g(x)f(y) \left[ \dfrac{1}{|x + r' - y|} - \dfrac{1}{|x + r'' - y|} \right]dUdV.
\end{equation}
Since $r'$ and $r''$ have the same direction, we now decompose $x$ and $y$ into two components, i.e., let
\begin{equation}
    \begin{cases}
        x_v = (x \cdot e_r) e_r \\
        x_h = x - x_v
    \end{cases} 
    \quad \text{and} \quad
    \begin{cases}
        y_v = (y \cdot e_r)e_r \\
        y_h = y - y_v
    \end{cases}
\end{equation}

Recall the assumption that $f(x) = 0$ outside the ball $B_{r}(x_0)$, then we have
\begin{equation}
\begin{split}
        |x + r' - y|^2 = & |x_h + y_h|^2 + |x_v + r' - y_v|^2 
        \\ =& |x_h + y_h|^2 + |x_v \cdot e_r + r' \cdot e_r - y_v \cdot e_r|^2, \\
        |x + r'' - y|^2 =& |x_h + y_h|^2 + |x_v + r'' - y_v|^2 \\
        =& |x_h + y_h|^2 + |x_v \cdot e_r + r''  \cdot e_r - y_v \cdot e_r|^2,
\end{split}
\end{equation}
where $e_r$ is the unit vector along $r'$.
Then consider the following two cases:
\begin{enumerate}
    \item $x_v \cdot e_r + r' \cdot e \ge y_v \cdot e_r$. 
    Since $r'' \cdot e_r > r' \cdot e_r$, we have $|x + r' - y|^2 < |x + r'' - y|^2$
    
    \item $x_v \cdot e_r + r' \cdot e < y_v \cdot e_r$.
    Since that $y_v \cdot e_r < r_0, x_v \cdot e_r + r' \cdot e_r \ge r_0 - \epsilon / 2$, and $x_v \cdot e_r + r'' \cdot e_r > r_0 + \epsilon / 2$,
    we have
    \begin{equation}
        (x_v \cdot e_r + r'' \cdot e_r) - y > \epsilon / 2 > y_v\cdot e_r - (x_v \cdot e_r + r' \cdot e_r),
    \end{equation}
    so that $|x + r' - y|^2 < |x + r'' - y|^2$.
\end{enumerate}
In both cases, $|x + r' - y| < |x + r'' - y|$ for any $x \in B_{\epsilon / 4}(0)$ and $y \in B_{r_0}(x_0)$, and thus
\begin{equation}
    \dfrac{1}{|x + r' - y|} - \dfrac{1}{|x + r'' - y|} > 0.
\end{equation}
This leads to $\text{PE}(f') < \text{PE}(f)$, which implies that as long as the distribution $f$ has zero mass on the border region of $\Omega$, i.e., $\Delta_\epsilon \Omega$, we can always find a new distribution $f'$ having lower potential energy by moving some mass of $f$ to the border region.
Thus, the distribution that minimizes potential energy must have a non-zero probability mass in the border, which completes our proof.
\end{proof}

\section{Detailed Experiment Settings}

The description of network structures is in \Cref{tab:model-arch}.
For all tasks, including the original tasks and fine-tuning attacks, we use the Adam optimizer with default parameters for faster training.
\subsubsection{Original Tasks.}
For training with PELoss, DcorLoss, and LabelDP we train the model for 100 epochs and save the model with the best performance on the validation set.
For PELoss and DcorLoss, We save the best model between the 90th to 100th epoch to ensure sufficient optimization of the loss function.
For LabelDP, we save the best model between the 50th to 100th epoch to avoid severe overfitting in more epochs.
As for vanilla split training, we use the early stopping strategy of 20 epochs.

\subsubsection{Fine-tuning Attack.}
During the fine-tuning, the entire bottom model remains fixed, as we empirically find this has the best attack accuracy.
The stop criterion for the attacker's fine-tuning tasks is that the training error is smaller than 0.01 or the epochs reach 1,000.
When the top model contains only one fully-connect layer, we use the mean of forward embeddings of each class as the initialization of the parameters of the top model, which greatly accelerates the fine-tuning attack.

\begin{figure}[ht]
    \centering
    \includegraphics[width=1\linewidth]{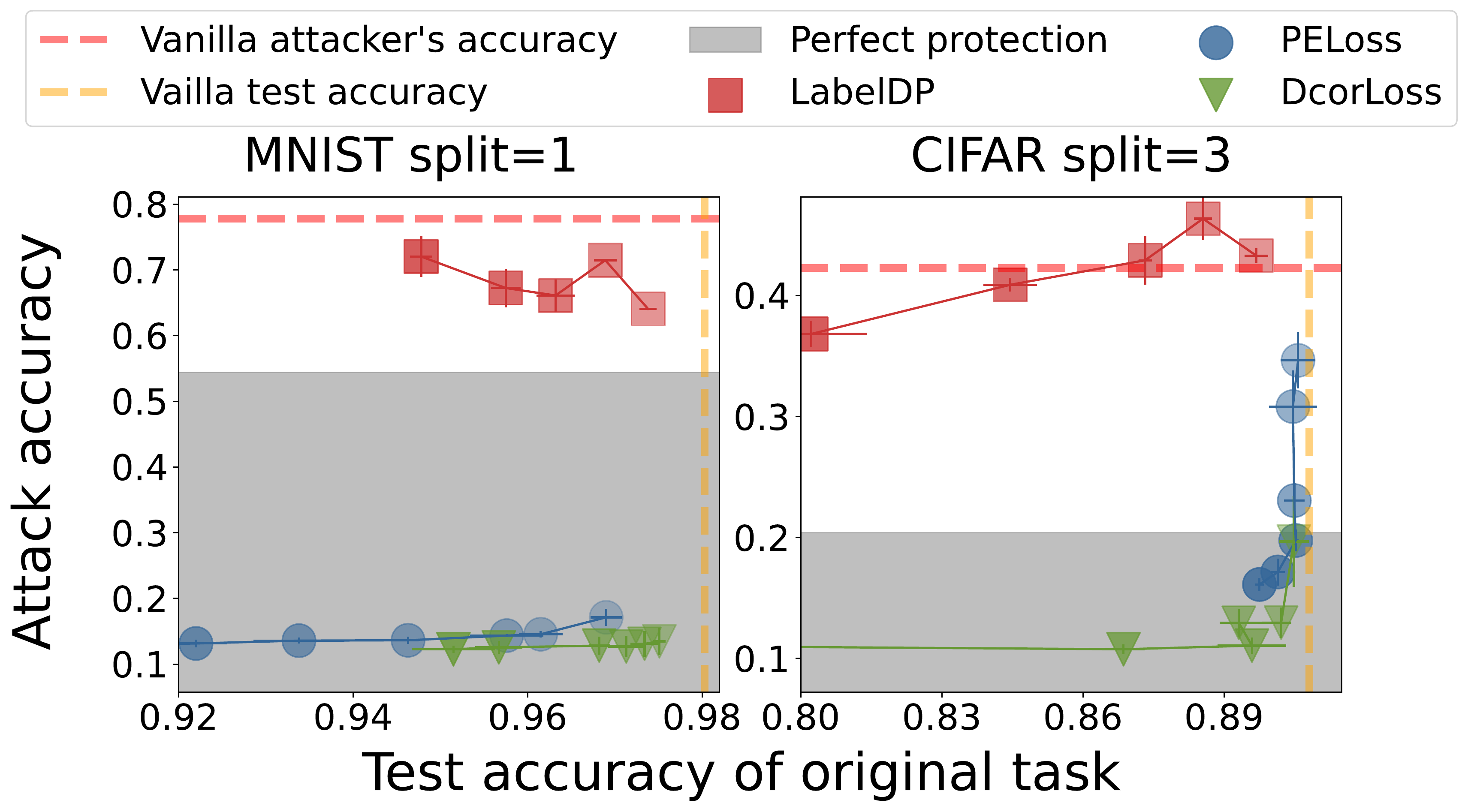}
    \caption{Test accuracy of original task vs. accuracy of clustering attack.}
    \label{fig:clustering-middle}
\end{figure}

\subsubsection{Clustering Attack.}
For clustering attacks, we adopt the k-Means algorithm from the Scikit-Learn library with default parameters, i.e., 100 rounds and 10 random initialization points.
The clustering result is measured as the maximum accuracy under all label assignments, since the clustered label may not correspond to the original label.
For example, suppose $c_{i,j}$ is the number of samples in $i$-th cluster with ground truth label $j$, where $1 < i, j \le C$, the accuracy is defined as
\begin{equation}
    \max\limits_{k_1,\cdots,k_C \in \text{Permutation}(1,\cdots,C)} \dfrac{\sum_{i=1}^C c_{i,k_i}}{\sum_{i,j=1}^C c_{i,j}}.
\end{equation}

\section{More Experiment Results}
\subsection{Other Split Position}
To study the effect of different split positions, we conduct additional experiments on MNIST and CIFAR.
In MNIST, we split the DNN by its second layer instead of its third (last) layer, which means the bottom model is a single fully-connect layer and the top model contains two fully-connect layers.
In CIFAR, we split the ResNet-20 model by its third residual block, which means the bottom model contains the first convolutional layer and the first two residual blocks, and the top model includes the third residual block, the final pooling and fully-connect layer.
The result of fine-tuning attacks and clustering attacks are displayed in \Cref{fig:fine-tuning-middle,fig:clustering-middle}.

We can see that splitting at previous layers improves the performance of all methods.
However, PELoss still outperforms DcorLoss against fine-tuning attacks and can achieve perfect protection even with 32 leaked labels per class.
As for clustering attacks, both methods can achieve perfect protection.

\begin{figure}[h]
    \centering
    \includegraphics[width=1\linewidth]{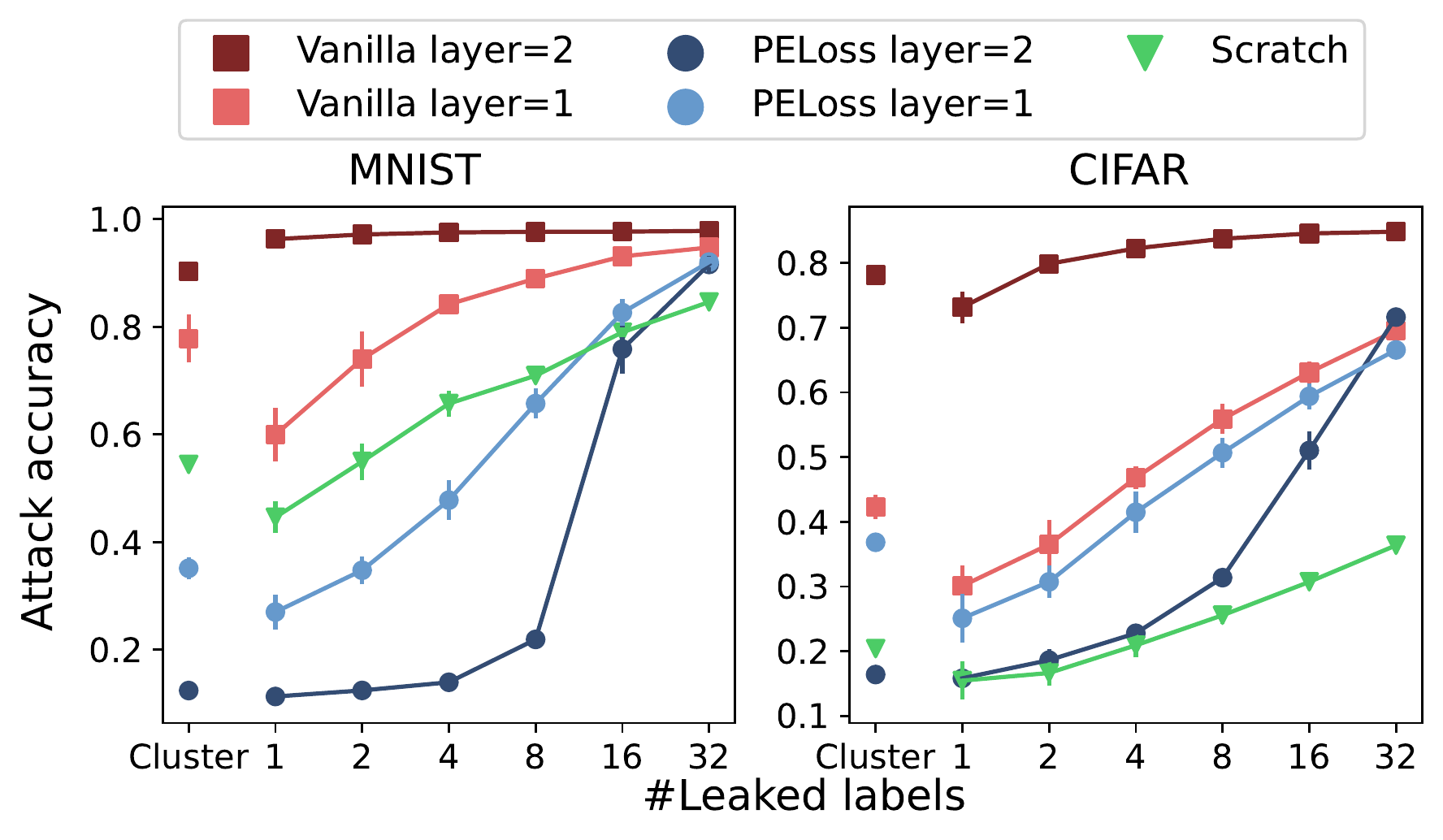}
    \caption{Test accuracy of original task vs. accuracy of clustering attack when attacking on earlier layers.
    Dots in the first column is the result of clustering attacks.
    For PELoss, $\alpha=4$.}
    \label{fig:attack-layer}
\end{figure}

\subsection{White-Box Attack at Earlier Layer}
While in previous experiments, we assume the attacker only gets the output of the bottom model.
If the attacker has full access to the bottom model, he can also get embeddings from the outputs of earlier layers.
Hence, we report the attack results on earlier layers of the MNIST and CIFAR datasets.
The split position is the last linear layer, while the attack position is the first layer and third layer for MNIST and CIFAR.
The detailed definition of the layer can be found in the previous section.
The results are shown in \Cref{fig:attack-layer}.

Although the attacks on earlier layers show better performance, it is still quite lower than the non-protected case (vanilla split learning).

\begin{figure}[h]
    \centering
    \includegraphics[width=1\linewidth]{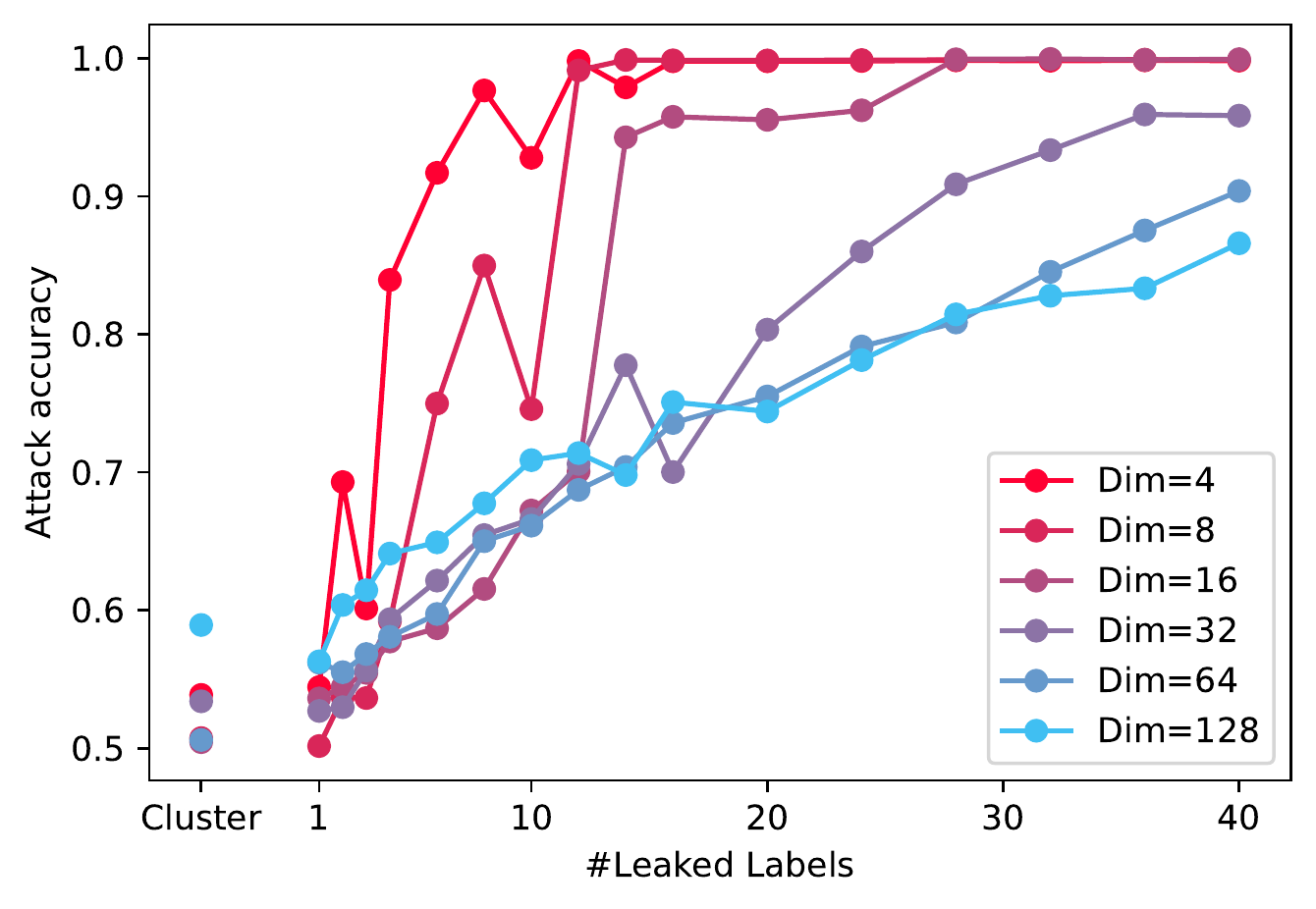}
    \caption{Test accuracy of original task vs. accuracy of clustering attack of different forward embedding dimensions on MNIST (2 class).
    Dots in the first column is the result of clustering attacks.
    The loss coefficient $\alpha = 1$.}
    \label{fig:mnist-subclass2}
\end{figure}
\subsection{Dimension of the Forward Embedding}
To study how the dimensions of the forward embedding affect the attack accuracy, we report the result of different forward embedding dimensions on the MNIST dataset with two labels (we only keep the `0' and `1' samples) in \Cref{fig:mnist-subclass2}.

We can see that a larger dimension on the forward embeddings generally makes the fine-tuning attacks harder.
For example, when the dimension is 4, the attacker can achieve near $100\%$ accuracy with 10 labeled samples in each class.
In contrast, when the dimension is 128, the attacker can only achieve an accuracy of $<90\%$ even with 40 labeled samples in each class.
This is quite intuitive, since we need more parameters to classify data of higher dimension, and naturally, more samples are needed.
However, when there are very few labeled samples (e.g., $1\sim 5$), or for clustering attacks, high dimensions seem to be beneficial to the attacker, as the attack accuracy is slightly higher than in low-dimension cases.
This counterintuitive phenomenon shall be studied in the future.

\begin{figure*}[h]
    \centering
    \includegraphics[width=\linewidth]{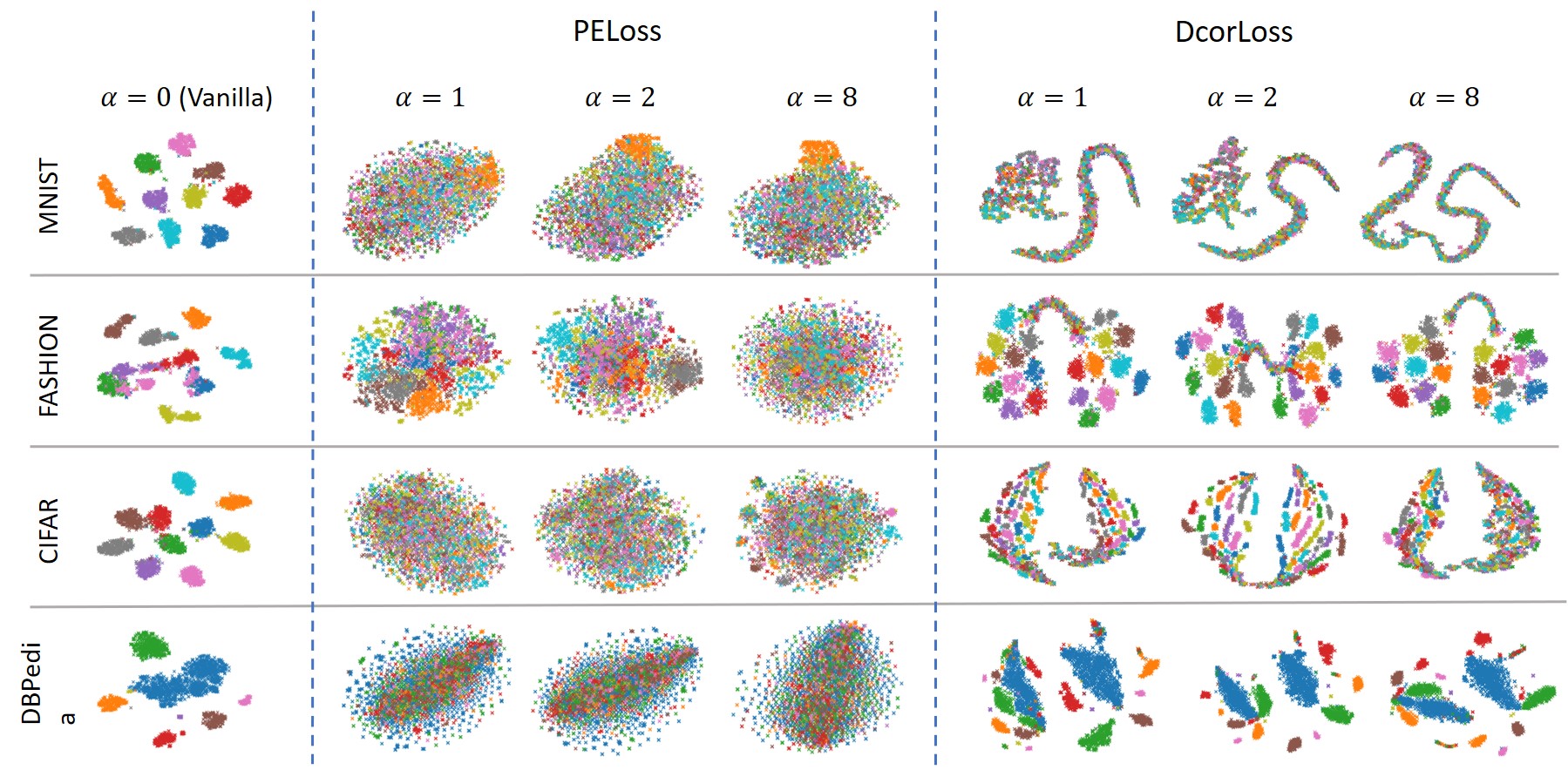}
    \caption{t-SNE of the forward embeddings.
        Different color represents the forward embedding of different classes.}
    \label{fig:tsne}
\end{figure*}
\subsection{Visualization of Forward Embedding}
To directly show the effect of potential energy loss, we visualize the bottom model output using t-SNE~\cite{van2008tsne} in \Cref{fig:tsne}.
We can see that without PELoss, the embeddings of different classes are clustered with a relatively large margin.
Adding PELoss quickly makes the embeddings of different classes all mixed together and seem to be completely random.
On the other side, DcorLoss smashes embeddings to some extent, but it seems like dividing the same-class embeddings into multiple clusters, which still preserves much information.

\end{document}